\documentclass[11pt]{article}
\usepackage{KLMM/klmm}
\usepackage{xurl}
\newcommand{\Z}{\mathbb{Z}}
\newcommand{\Q}{\mathbb{Q}}
\newcommand{\one}{\mathbf{1}}
\newcommand{\NN}{\mathbb{N}}

\newcommand{\tr}{\operatorname{tr}}
\newcommand{\sgn}{\operatorname{sgn}}
\title{Exact rational certificates and a Lean formalization of the eleven-variable Vasc inequality}
\author{Dakai~Guo\thanks{State Key Laboratory of Mathematical Sciences, Academy of Mathematics and Systems Science, Chinese Academy of Sciences, and University of Chinese Academy of Sciences. This work is supported by the Strategic Priority Research Program of the Chinese Academy of Sciences under Grants XDA0480502 and XDA0480503.}
    \and Ruichen Qiu\samethanks
    \and Yichuan Cao\samethanks
    \and Ruyong Feng\samethanks
    \and Xiao-Shan Gao\samethanks
}
\date{\today}
\begin{document}
\maketitle
\begin{abstract}
We prove that the cyclic sum
$\sum_{i=0}^{10}(x_i-x_{i+1})/(x_{i+1}+x_{i+2})$ is nonnegative for every
strictly positive real eleven-tuple, with indices taken modulo eleven.
Clearing denominators reduces the problem to a homogeneous polynomial.
Subtracting the smallest coordinate then separates the proof into
nonnegativity of a common-translation derivative and nonnegativity on a
full boundary face. We supply exact rational polynomial certificates in Gram-matrix form for both
statements. The derivative certificate uses a linear multiplier, centered
monomials and cyclic averaging; the boundary certificate uses a strictly
positive quadratic multiplier and uncentered monomials. Their verification
requires 6,069 cyclic coefficient comparisons and 34,304 literal coefficient
comparisons. All 606 certificate matrices are certified by integer
congruences with strictly diagonally dominant images. A Lean 4 development
formally verifies the certificate checks and their connection to the original
inequality. The exact certificate data, replay programs, and formal sources
are publicly available.
\end{abstract}


\section{Introduction}\label{sec:intro}
We prove the eleven-variable Vasc inequality using two exact rational
polynomial certificates and formalize the proof in Lean.
For an eleven-tuple $x=(x_0,\ldots,x_{10})$ of positive real numbers, define
\begin{equation}\label{eq:cyclic}
 C_{11}(x)=\sum_{i=0}^{10}
 \frac{x_i-x_{i+1}}{x_{i+1}+x_{i+2}},
\end{equation}
where every subscript is read modulo eleven.
\begin{theorem}\label{thm:main}
For every $x\in\R_{>0}^{11}$, one has $C_{11}(x)\geq0$.
\end{theorem}
The cyclic orientation is as displayed; the positive coordinates may have
arbitrary order and ratios.

The proof combines a short analytic argument with finite algebraic checks.
After clearing denominators, let $P$ denote the resulting degree-eleven
polynomial. We certify its derivative in the direction $\one=(1,\ldots,1)$
and its restriction to the face $x_{10}=0$. Subtracting the minimum
coordinate moves an arbitrary positive tuple to a boundary face, and the
certified derivative controls the return along a common translation.
This separates the domain argument from the construction of the two
polynomial certificates.

For the derivative we use centered monomial vectors, a linear
multiplier and the cyclic group of order eleven. For the boundary we use
ordinary monomial vectors and a quadratic multiplier that combines positive
monomial coefficients with a positive-definite quadratic form. Every matrix
in both certificates admits a supplied integer congruence whose image is
strictly diagonally dominant. Consequently, the complete algebraic
verification consists of finite integer coefficient equalities and strict
integer inequalities.

The Lean formalization connects these data to a theorem
with precisely the hypotheses of Theorem~\ref{thm:main}. The formal
project proves the soundness of its matrix and sparse-polynomial checks,
checks the concrete certificates, and supplies the analytic assembly.
Its large finite computations use \texttt{native\_decide}, whose trust
boundary is described in Section~\ref{sec:trust}.

\subsection{Related work and contribution}
The family in~\eqref{eq:cyclic} is referred to as the Vasc inequality.
Recent finite-certificate work treats the same positive-variable convention
for nine variables \citep{GuoQiuCaoFeng2026}. Xu and Yu posted a preprint
proving the eleven-variable case
\citep{XuYu2026}. Xu and Yu's work and the present work were carried out independently.
Their construction uses $10!$ normalized rank regions,
cumulative gap variables, coefficient positivity,
arithmetic--geometric-mean circuits and subdivision covers. Our proof uses a
different organization: two global polynomial sign conditions, one for a
translation derivative and one for a full boundary face, together with
weighted Gram identities. We have also formalized the proof in Lean,
including the certificate checks and the analytic steps connecting them
to the original inequality. The exact certificate proof reported here was
completed by 29 September 2026.

Exact recovery from numerical sum-of-squares data is an established
symbolic--numeric technique. \citet{KaltofenLiYangZhi2012} study exact
rational-function sum-of-squares certification in polynomial optimization,
including rational recovery and matrix positivity checks.
\citet{MagronSafeyElDin2021} develop algorithms for several exact rational
positivity representations. The role of rational certificate size and
its dependence on representation is further examined by
\citet{DavisPapp2024}, using dual certificates for weighted sums of squares.
Our construction uses primal Gram certificates and requires no convergence
or complexity guarantees from those algorithms. Its contribution is the
explicit two-certificate construction for this inequality and the formal
connection to the real-valued theorem.

The construction addresses this specific inequality; multiplier
optimality, certificate minimality and equality classification remain
outside its scope. We report certificate sizes and replay timings for the
supplied instance, with timings measured in one recorded environment.

\section{Reduction to two polynomial sign conditions}\label{sec:reduction}
Put
\[
 a_i=x_i-x_{i+1},\qquad d_i=x_{i+1}+x_{i+2},\qquad
 D(x)=\prod_{i=0}^{10}d_i.
\]
Define the polynomials
\begin{align}
 P(x)&=\sum_{i=0}^{10}a_i\prod_{j\ne i}d_j,\label{eq:P}\\
 H(x)&=\sum_{r=0}^{10}\partial_rP(x),\label{eq:Hdef}\\
 B(z)&=P(z_0,\ldots,z_9,0),\qquad z\in\R^{10},\label{eq:Bdef}\\
 \ell(x)&=\sum_{i=0}^{10}x_i.\label{eq:ell}
\end{align}
The total degrees of $P$, $H$ and $B$ are respectively eleven, ten and eleven.
All boundary evaluations in this article concern these polynomials.
In particular, the definition of $B$ does not require a value for a
rational summand whose denominator vanishes.

Under the substitution $x\mapsto x+t\one$, the differences $a_i$ remain
unchanged and every $d_i$ becomes $d_i+2t$. Differentiating the finite
products in~\eqref{eq:P} gives the explicit expression
\begin{equation}\label{eq:H}
 H(x)=2\sum_{i=0}^{10}\sum_{j\ne i}
       a_i\prod_{k\notin\{i,j\}}d_k.
\end{equation}
This formula will also serve as the independently expanded target for the
derivative certificate. It involves $110$ ordered pairs $(i,j)$ and nine
binomial factors for each pair.

\begin{lemma}[Denominator clearing]\label{lem:denom}
For every $x\in\R_{>0}^{11}$, one has $D(x)>0$ and
\begin{equation}\label{eq:denom}
 P(x)=D(x)C_{11}(x).
\end{equation}
Consequently, $P(x)\geq0$ implies $C_{11}(x)\geq0$.
\end{lemma}
\begin{proof}
Every $d_i=x_{i+1}+x_{i+2}$ is positive, so $D(x)>0$. For each $i$,
\[
 D(x)\frac{a_i}{d_i}=a_i\prod_{j\ne i}d_j.
\]
Summing over $i$ gives~\eqref{eq:denom}. Since $D(x)>0$, division by
$D(x)$ shows that $P(x)\geq0$ implies $C_{11}(x)\geq0$.
\end{proof}

\begin{lemma}\label{lem:reduction}
Suppose that
\begin{equation}\label{eq:signs}
 H(x)\geq0\quad(x\in\R_{\geq0}^{11}),\qquad
 B(z)\geq0\quad(z\in\R_{\geq0}^{10}).
\end{equation}
Then Theorem~\ref{thm:main} holds.
\end{lemma}
\begin{proof}
Fix $x\in\R_{>0}^{11}$ and put $t=\min_i x_i$ and $y=x-t\one$.
Then $t>0$, all coordinates of $y$ are nonnegative, and at least one is
zero. Let $\rho$ denote the cyclic permutation of coordinates. The polynomial
$P$ is invariant under this permutation: shifting all indices reindexes both
the sum in~\eqref{eq:P} and each product with one excluded index. A suitable
power of $\rho$ therefore places a zero coordinate of $y$ in position ten
without changing $P(y)$. The boundary assumption gives
$P(y)\geq0$.

For $0\leq s\leq t$, every coordinate of $y+s\one$ is nonnegative.
The chain rule and~\eqref{eq:Hdef} give
\[
 \frac{d}{ds}P(y+s\one)=H(y+s\one)\geq0.
\]
Integration of this polynomial derivative yields
\[
 P(x)-P(y)=\int_0^t H(y+s\one)\,ds\geq0.
\]
Hence $P(x)\geq0$. Lemma~\ref{lem:denom} now gives
$C_{11}(x)\geq0$.
\end{proof}

The boundary certificate in Lemma~\ref{lem:reduction} covers the entire
closed face: the minimum can occur at several coordinates, so the
translated tuple may lie on a further subface. Positivity of the original
denominators is used only after returning to the positive tuple.

\section{An exact certificate language}\label{sec:language}
For $e=(e_0,\ldots,e_{n-1})\in\NN^n$, write
$x^e=\prod_i x_i^{e_i}$ and $|e|=\sum_i e_i$. The convention for evaluating
monomials is $x_i^0=1$, including at $x_i=0$.
The certificates have terms of the form
\[
 x^m V(x)^TQ V(x),
\]
where $m$ is a binary exponent vector, $V$ is a vector of monomials or
monomial differences, and $Q$ is a rational symmetric matrix. On the
nonnegative orthant, the monomial factor $x^m$ is nonnegative. Thus it
suffices to establish the sign of the quadratic form.

\subsection{Triangular polynomial coordinates}
A $q\times q$ symmetric matrix is encoded by $q(q+1)/2$ rational numbers
$u_{ab}$, read in increasing column order $b$, and then in increasing row
order $a\geq b$. The interpretation is
\begin{equation}\label{eq:tri}
 Q_{aa}=u_{aa},\qquad Q_{ab}=Q_{ba}=\frac{u_{ab}}2\quad(a>b).
\end{equation}
Consequently,
\begin{equation}\label{eq:tri-form}
 V^TQV=\sum_{a\geq b}u_{ab}V_aV_b.
\end{equation}
The stored coordinates are coefficients of products in the quadratic
polynomial; off-diagonal coordinates are twice the corresponding entries
of $Q$.

\subsection{Integer congruence certificates}
The same sign test is used for every Gram matrix and for the matrix in the
boundary multiplier.
Positive definiteness is a convenient uniform certificate for the required
nonnegativity; no strict positivity of an individual Gram term is needed
below.
\begin{lemma}\label{lem:matrix}
Let $Q\in\Q^{q\times q}$ be symmetric. Suppose that $c>0$, that $A=cQ$
has integer entries, and that a square integer matrix $R$ satisfies
\begin{equation}\label{eq:dd}
 J=RAR^T,\qquad J_{ii}>\sum_{j\ne i}|J_{ij}|\quad(1\leq i\leq q).
\end{equation}
Then $Q$ is positive definite.
\end{lemma}
\begin{proof}
Set $\delta_i=J_{ii}-\sum_{j\ne i}|J_{ij}|$, so that $\delta_i>0$.
Expansion gives, for every $v\in\R^q$,
\begin{equation}\label{eq:dd-squares}
 v^TJv=\sum_i\delta_i v_i^2+
   \sum_{i<j:\,J_{ij}\ne0}|J_{ij}|
       \bigl(v_i+\sgn(J_{ij})v_j\bigr)^2.
\end{equation}
The coefficient of $v_i^2$ on the right is $J_{ii}$, and the coefficient
of $v_iv_j$ is $2J_{ij}$. A zero off-diagonal entry contributes zero.
If $v\ne0$, some term $\delta_i v_i^2$ is positive and the remaining
terms are nonnegative. Therefore $J$ is positive definite.

If $R^Tv=0$, then $v^TJv=(R^Tv)^TA(R^Tv)=0$, which forces $v=0$.
Thus the square linear map $R^T$ is injective and hence invertible.
For $w\ne0$, write $w=R^Tv$ with $v\ne0$. It follows that
$w^TQw=c^{-1}v^TJv>0$.
\end{proof}

Strict diagonal dominance also establishes the invertibility of the
supplied congruence matrix. Given $A$ and $R$, multiplication and the row
tests in~\eqref{eq:dd} take place entirely in $\Z$.

\section{The common-translation certificate}\label{sec:translation}
Let $\rho$ denote a cyclic rotation of the eleven variables, acting on the
whole polynomial to which it is applied. The derivative certificate is
\begin{equation}\label{eq:trans-cert}
 \ell(x)H(x)=\sum_{(m,E)}\sum_{r=0}^{10}
 \rho^r\!\left[x^m V_E(x)^TQ_{m,E}V_E(x)\right].
\end{equation}
The outer sum contains $93$ specified blocks. In each block,
$m\in\{0,1\}^{11}$ and
$E=(e_0,\ldots,e_q)$ is a list of distinct nonnegative exponent vectors
satisfying
\[
 |m|+2|e_a|=11\qquad(0\leq a\leq q).
\]
The centered vector in this block is
\begin{equation}\label{eq:centered}
 V_E(x)=\bigl(x^{e_1}-x^{e_0},\ldots,x^{e_q}-x^{e_0}\bigr)^T.
\end{equation}
These centered vectors form the basis of the derivative certificate.

The data contain $248,264$ triangular polynomial coordinates. Reading them
block by block defines every $Q_{m,E}$. The matrix dimensions sum to
$5,975$, and the largest dimension is $184$. Each block also has a
positive integer \texttt{qd} and a square integer matrix $R$ such that
$A=\texttt{qd}\,Q$ is integral and~\eqref{eq:dd} holds. The exact
verification checks all $5,975$ strict row inequalities for the supplied
arrays. Thus
Lemma~\ref{lem:matrix} proves positive definiteness of all $93$ matrices.
A further positive integer \texttt{rd} in the data specifies the scale of
the corresponding rational congruence $R/\texttt{rd}$; this scale is not
needed for the sign test.

\subsection{Cyclic coefficient comparison}\label{sec:orbits}
The common denominator of the stored triangular coordinates is
\[
L=2\,000\,000\,000\,000.
\]
Every stored coordinate is an integer multiple of $1/L$. Thus multiplying
the certificate identity by $L$ gives integer coefficients: $G_t$ is $L$
times the certificate-side coefficient, while $Lf_t$ is $L$ times the
independently expanded target coefficient.
For an exponent $e$, let $[e]$ be the lexicographically least element of
its cyclic orbit. For degree-eleven exponents in eleven variables, the
only orbit with a nontrivial stabilizer is that of $\one$. Indeed, any
nonzero cyclic shift generates the group of order eleven. An exponent
fixed by such a shift has all coordinates equal, and its total degree
then forces every coordinate to be one. Hence define
\[
 s_e=\begin{cases}11,&e=\one,\\1,&e\ne\one.\end{cases}
\]
If $f=\sum_e c_e x^e$ and $t$ is an orbit representative, then
\begin{equation}\label{eq:orbit}
 [x^t]\sum_{r=0}^{10}\rho^r f
       =s_t\sum_{[e]=t}c_e.
\end{equation}
For each exponent in the orbit of $t$, exactly $s_t$ rotations send it to
$t$, which proves this multiplicity formula.

By~\eqref{eq:tri-form}, a stored coordinate $u_{ab}$ contributes
$u_{ab}x^mV_aV_b$. Its four monomials are
\begin{equation}\label{eq:four}
 x^{m+e_a+e_b}-x^{m+e_a+e_0}
 -x^{m+e_0+e_b}+x^{m+2e_0}.
\end{equation}
Multiplying their coefficients by $L$ gives the integer weights
\[Lu_{ab},\quad-Lu_{ab},\quad-Lu_{ab},\quad Lu_{ab}.\]
Both middle contributions must be
retained when $a=b$. Replace every exponent by its cyclic representative,
multiply its contribution by the stabilizer factor, and sum over all
coordinates of all blocks. Let the resulting integer at representative
$t$ be $G_t$.

Independently expand~\eqref{eq:H}, using the $110$ ordered pairs, and
multiply by $\ell$. Let $f_t$ be the literal coefficient of $x^t$ in this
expansion. The frame contains $6,069$ distinct representatives. The
checking procedure verifies that every exponent generated on either side
belongs to one of these orbits and that target coefficients agree
throughout each orbit. It then checks
\begin{equation}\label{eq:trans-rows}
 G_t=Lf_t
\end{equation}
for every listed representative, including those with zero coefficients.
All comparisons hold exactly for the supplied data. Both sides
of~\eqref{eq:trans-cert} are cyclically invariant. Equation~\eqref{eq:orbit}
and the support coverage check therefore prove equality at every exponent;
outside the checked orbits both sides have zero coefficient. This
establishes the polynomial identity~\eqref{eq:trans-cert}.

\subsection{The derivative sign on the closed orthant}
\begin{proposition}\label{prop:H}
For every $x\in\R_{\geq0}^{11}$, one has $H(x)\geq0$.
\end{proposition}
\begin{proof}
Every monomial factor $x^m$ in~\eqref{eq:trans-cert} is nonnegative, and
every quadratic form is nonnegative by the matrix certificates. These
properties are preserved by cyclic rotation. Thus $\ell(x)H(x)\geq0$.
If $x\ne0$, then $\ell(x)>0$ and division gives $H(x)\geq0$. At the
origin, formula~\eqref{eq:H} has positive degree and gives $H(0)=0$.
\end{proof}

\section{The full boundary certificate}\label{sec:boundary}
The boundary certificate uses ten independent nonnegative variables
$z=(z_0,\ldots,z_9)$. Set $N=10^{14}$, and let $\varepsilon_i$ denote the
$i$th unit vector in ten coordinates. Enumerate the $55$ pairs $(i,j)$,
$0\leq i\leq j\leq9$, in lexicographic order, and associate to each pair
the exponent $\alpha=\varepsilon_i+\varepsilon_j$.
The supplied data give positive integer numerators $p_\alpha$ and
integer triangular coordinates $u^W_\alpha$. The matrix $W$ is defined
using~\eqref{eq:tri} with coordinates $u^W_\alpha/N$. The pair $(i,j)$
corresponds to column $i$, row $j$, so this pair order agrees with the
triangular storage convention. Define
\begin{equation}\label{eq:M}
 M(z)=\frac1N\sum_\alpha p_\alpha z^\alpha+z^TWz.
\end{equation}
All $55$ weights are strictly positive; the smallest numerator is
$45\,839\,617\,787$. The data also satisfy the exact normalization
\begin{equation}\label{eq:norm}
 \frac1N\sum_\alpha p_\alpha+\tr(W)=1.
\end{equation}
The normalization fixes the scale of the supplied multiplier. The sign
argument below uses positivity, not this choice of scale.

The boundary identity is
\begin{equation}\label{eq:boundary-cert}
 M(z)B(z)=\sum_{(m,E)}z^m Z_E(z)^TQ_{m,E}Z_E(z).
\end{equation}
There are $512$ blocks. Here $m\in\{0,1\}^{10}$,
$E=(e_0,\ldots,e_{q-1})$ consists of distinct nonnegative exponent vectors,
and
\[
 |m|+2|e_a|=13,\qquad
 Z_E(z)=\bigl(z^{e_0},\ldots,z^{e_{q-1}}\bigr)^T.
\]
The boundary certificate uses uncentered vectors and literal ten-variable
expressions, without cyclic averaging.

The data contain $1,235,364$ integer triangular coordinate numerators for
the $512$ Gram matrices. Dividing by $N$ and
applying~\eqref{eq:tri} determines the matrices, whose maximum dimension
is $115$. The multiplier matrix $W$ has dimension ten. For each of these
$513$ matrices, take $A=2NQ$, or $A=2NW$ for the multiplier. Its diagonal
entries are twice the stored integer numerators, and its off-diagonal
entries are the stored numerators themselves. The certificate supplies
one matching square integer matrix $R$ for each $A$, with the multiplier
matrix first and the Gram blocks in their prescribed order. All $34,314$
strict row inequalities in~\eqref{eq:dd} hold for the supplied arrays. Thus
$W$ and all boundary
Gram matrices are positive definite.

\subsection{Literal coefficient comparison}
Expand $B$ directly from~\eqref{eq:P}, setting $x_{10}=0$ and $x_i=z_i$
for $i<10$. For each $i$, multiply the ten binomials whose indices differ
from $i$, multiply by $x_i-x_{i+1}$, and sum. Terms containing $x_{10}$
vanish. This yields $2,302$ nonzero monomials of degree eleven. The
construction uses the defining finite products rather than a trusted
cache of boundary coefficients.

By the triangular convention, the coefficient of $z^\alpha$ in $z^TWz$
is $u^W_\alpha/N$ for every diagonal or off-diagonal pair $\alpha$.
Therefore
\begin{equation}\label{eq:boundary-int}
 NM(z)B(z)=B(z)\sum_\alpha
                  (p_\alpha+u^W_\alpha)z^\alpha.
\end{equation}
Let $\beta_e$ be the integer coefficient of $z^e$ in the right-hand side
of~\eqref{eq:boundary-int}. On the Gram side, each triangular numerator
$u_{ab}$ contributes $u_{ab}$ at exponent $m+e_a+e_b$; for an off-diagonal
entry, the factor $2$ in the quadratic form cancels the factor $1/2$ in
~\eqref{eq:tri}. Let $\gamma_e$
be the sum of these integer contributions over all blocks. These are
precisely the coefficients of $N$ times the right-hand side
of~\eqref{eq:boundary-cert}.

The frame lists $34,304$ distinct degree-thirteen exponents. Every
exponent generated on either side lies in this list, and the check
verifies $\beta_e=\gamma_e$ at every listed exponent, including zero
rows. All these equalities hold. Outside the list both coefficients
vanish. Unlike the derivative comparison, this is a literal coefficient
comparison without orbit representatives or stabilizer factors.
Uniqueness of polynomial coefficients proves~\eqref{eq:boundary-cert}.

\subsection{Strict positivity of the multiplier}
\begin{proposition}\label{prop:B}
For every $z\in\R_{\geq0}^{10}$, one has $B(z)\geq0$.
\end{proposition}
\begin{proof}
Each term on the right-hand side of~\eqref{eq:boundary-cert} is
nonnegative: $z^m\geq0$ and the Gram matrix is positive definite. Hence
$M(z)B(z)\geq0$. The quadratic form $z^TWz$ is nonnegative, and every
weighted monomial in~\eqref{eq:M} is nonnegative. If $z\ne0$, choose an
index $i$ with $z_i>0$. The diagonal weight
$p_{\varepsilon_i+\varepsilon_i}$ is strictly positive, so the term
$(p_{\varepsilon_i+\varepsilon_i}/N)z_i^2$ is strictly positive. Thus
$M(z)>0$, and
division gives $B(z)\geq0$. At the origin, the degree-eleven polynomial
$B$ vanishes.
\end{proof}
The proof covers all further zero coordinates, regardless of which
original denominator factors vanish there. It does not use
Proposition~\ref{prop:H}; the two sign conditions are certified
independently.

\begin{proof}[Proof of Theorem~\ref{thm:main}]
Proposition~\ref{prop:H} proves the derivative hypothesis in
Lemma~\ref{lem:reduction}, and Proposition~\ref{prop:B} proves its boundary
hypothesis. The lemma therefore gives $P(x)\geq0$ for every positive tuple;
Lemma~\ref{lem:denom} then gives $C_{11}(x)\geq0$.
\end{proof}

\section{Construction and finite verification}\label{sec:computation}
The final proof objects consist of six JSON arrays: a monomial frame,
rational matrix data and integer congruence data for each of the two
identities. Their interpretation is fixed by
Sections~\ref{sec:language}--\ref{sec:boundary}. Once these arrays are
available, no numerical optimization is needed to verify the proof.

\subsection{From numerical candidates to exact data}
Semidefinite optimization provides candidate Gram coordinates, followed by
rational reconstruction and exact coefficient repair. For the derivative,
the recorded reconstruction has common coordinate denominator
$L=2\cdot10^{12}$. For the boundary, the
reconstruction uses denominator $N=10^{14}$ and enforces the multiplier
normalization~\eqref{eq:norm}.

The boundary multiplier combines a positive-coefficient quadratic part
with a positive-semidefinite quadratic-form part. The representation
in~\eqref{eq:M} gives flexibility in the search while retaining a direct
strict-positivity argument on the nonzero closed orthant. The resulting
$W$ and Gram matrices are in fact positive definite, as established by
their exact congruence data.

The two reconstruction routines exploit different coefficient structures.
For the derivative, the coordinate-to-coefficient map is assembled as a
sparse integer matrix. After initial rounding, an exact rational
correction is solved on selected columns, and the full coefficient residual
is recomputed. For the boundary, each uncentered triangular coordinate
contributes to a single exponent. After fixing the multiplier normalization,
the integer residual in each coefficient row is distributed among the
coordinates contributing to that row. The recorded boundary correction
changes any one Gram numerator by at most two units at denominator $10^{14}$.

To construct an integer congruence, the implementation computes a
Cholesky factor of a floating-point approximation to the candidate
rational matrix, approximates the inverse factor, rescales and rounds the
entries to integers. The translation construction uses scale $10^8$ and
the boundary construction uses scale $10^{10}$ for this last step. The
resulting matrix $R$ is then tested against the exact integral matrix $A$.
The exact inequalities~\eqref{eq:dd} certify the resulting $R$.
No Cholesky factorization is needed during replay.

Coefficient repair and positive-definiteness certification are distinct
steps. A repaired rational solution is accepted only when the complete
integer coefficient comparison succeeds and every congruence passes the
strict row test. A small numerical residual does not establish a polynomial
identity, and an approximate spectral margin does not establish matrix
positivity. The replay programs check the delivered rational arrays
independently of the optimizer's numerical state.

\subsection{A complete checking procedure}
The checking procedure takes as input the six arrays and the
formulas~\eqref{eq:P} and~\eqref{eq:H}.
All coefficient arithmetic is integral after multiplication by the
specified denominators.
\begin{enumerate}
\item Parse each frame and check its dimensions, binary masks, exponent
nonnegativity, distinctness and degree conditions. Consume the coordinate
arrays exactly in the prescribed block order, and reject missing or
unused coordinates.
\item Reconstruct each symmetric matrix according to~\eqref{eq:tri}.
Check the stated positive scaling factor and integrality of $A$; check
the shape of its matching integer $R$; form $RAR^T$ and test every row
in~\eqref{eq:dd}.
\item Expand $\ell H$ from the defining binomial products. Expand every
centered Gram coordinate using~\eqref{eq:four}, accumulate cyclic
representatives with the stabilizer factor, and check both support
coverage and every equation~\eqref{eq:trans-rows}.
\item Expand $B=P(z,0)$ from its definition. Form the integer polynomial
in~\eqref{eq:boundary-int} and the integer Gram expansion. Check support
coverage and every one of the $34,304$ literal coefficient equalities.
\item Check all $55$ multiplier weights are positive and check the exact
normalization~\eqref{eq:norm}. Report success only if every preceding
check has completed.
\end{enumerate}
In the third and fourth steps, support coverage ensures that every
generated term occurs in the coefficient list. The procedures exhaust
the finite arrays, matrix rows and coefficient sets.

\begin{table}[tb]
\centering\small
\caption{Sizes of the exact certificates. Boundary matrix counts and row
counts include the ten-dimensional multiplier matrix $W$.}
\label{tab:sizes}
\begin{tabular}{@{}lrr@{}}
\toprule
Quantity & Translation & Boundary\\
\midrule
Variables & 11 & 10\\
Target identity degree & 11 & 13\\
Gram blocks & 93 & 512\\
Largest Gram dimension & 184 & 115\\
Gram triangular coordinates & 248,264 & 1,235,364\\
Matrices with congruence certificates & 93 & 513\\
Strict integer row inequalities & 5,975 & 34,314\\
Coefficient comparisons & 6,069 & 34,304\\
Coefficient organization & Cyclic orbits & Literal exponents\\
\bottomrule
\end{tabular}
\end{table}

The six JSON certificate arrays occupy $53,960,275$ bytes in total.
The formal project contains $2,031$ Lean source files, including the
generated certificate data.

Cyclic averaging is used for $\ell H$ because that target is
cyclically invariant in eleven variables. The boundary identity retains
the literal ten-variable polynomial and uses no cyclic quotient. The
matrix test, coordinate convention and integer coefficient arithmetic are
shared across both parts.

\section{Formalization in Lean}\label{sec:lean}
The formalization uses Lean 4 \citep{deMouraUllrich2021Lean4} and
Mathlib \citep{mathlib2020}. The project exports the following statement, using
\texttt{Fin 11} for the cyclic index set and the real numbers for values:
\begin{quote}\small
\begin{verbatim}
theorem vasc11_inequality (x : Fin 11 -> Real)
    (hx : forall i, 0 < x i) : 0 <= c11 x
\end{verbatim}
\end{quote}
The exported declaration is \texttt{VascN11.vasc11\_inequality}.
Its only hypothesis is strict positivity of the coordinates. Certificate,
derivative and boundary assumptions occur in intermediate interfaces and
are discharged in the final theorem.
In the Lean definitions, the derivative polynomial is introduced as the
derivative at zero of the common-translation polynomial. The target-link
identity proves that this polynomial is equal to the displayed
$H=\sum_r\partial_rP$ used in the article.

\subsection{Semantic layers}
The formal development separates the following tasks.
First, the matrix layer interprets triangular integer arrays and proves
that the checked congruence conditions imply nonnegativity of the
associated real quadratic forms. This layer also accounts for array sizes
and the conversion from triangular polynomial coefficients to symmetric
matrix entries. Its proof uses matrix results from Mathlib to establish
the strict diagonal-dominance implication. The direct square expansion
in Lemma~\ref{lem:matrix} provides the article's self-contained account of
the same criterion.

Second, the sparse-polynomial layer represents a polynomial as a list of
pairs consisting of an exponent vector and an integer coefficient.
Evaluation is defined by the finite sum of the corresponding monomials.
Addition concatenates lists; multiplication forms pairwise exponent sums
and coefficient products. Normalization sorts exponents, combines equal
exponents and removes zero coefficients. The soundness results connect
these list operations and normalization to evaluation over $\R$.
A successful normalized equality check therefore implies equality of
real evaluations.

Third, the polynomial interfaces connect the certificate expressions to
the literal $\ell H$ and $MB$ used in this article. Target expressions are
built from the defining differences, denominator factors, excluded-index
products and boundary substitution, tying each checked identity to the
polynomial whose sign is required.

Finally, the analytic assembly formalizes cyclic invariance, subtraction
of a minimum coordinate, nonnegativity on a face, monotonicity under common
translation and clearing the positive denominators. The formal proof uses
monotonicity from a nonnegative derivative; this supplies the same
comparison as the polynomial integration in
Lemma~\ref{lem:reduction}.

\subsection{Checked caches and staged composition}
Large polynomial expressions are checked in stages. A leaf check compares
each Gram block's compact polynomial representation with the expansion of
the underlying block. The derivative target separately checks its cached
linear factor and derivative polynomial; the boundary target checks its
cached multiplier and boundary polynomial. Aggregate checks compare
products of these checked target factors with sums of the checked leaves.

The staged interface proves that the combination of leaf checks, matrix
checks, target checks and aggregate identity checks implies the required
real sign statement. In the derivative sum, the cyclic rotations are
included before the final equality comparison. Balanced addition and
normalization organize the large finite expressions; their semantic
correctness is proved. Thus using a cache changes the organization of the
computation while preserving the connection to the original target.

The conversion from JSON arrays to Lean declarations is recorded by
provenance manifests, but those converters are not runtime dependencies of
the Lean project. The checked-in Lean sources contain the concrete data,
and Lean recomputes the certificate conditions from those sources.
Consequently, successful compilation does not depend on the presence of
the original Python optimizer or on a previous Python replay report.

\subsection{Trust boundary}\label{sec:trust}
The finite integer and polynomial checks use \texttt{native\_decide}.
The dependency audit for the final theorem lists the base axioms
\texttt{propext}, \texttt{Classical.choice} and \texttt{Quot.sound},
together with $1,215$ generated native-computation axioms.
It reports no \texttt{sorryAx} dependency or unrelated custom axiom.
The result therefore trusts the native evaluation/compiler mechanism in
addition to the usual Lean kernel and base axiom system. The real-valued
soundness lemmas and analytic
assembly are ordinary Lean proofs.

The large exact checks are carried out by compiled code rather than
reduced solely by the kernel. The project includes an axiom inspection
command to check these dependencies for the distributed theorem.

\section{Reproducibility and computational evidence}\label{sec:reproduce}
The supplementary materials have two entry points. The
exact replay bundle contains the six JSON certificate arrays and two
standalone Python programs. The Lean project contains the embedded data,
formal soundness proofs, concrete checks and final theorem. Both use the
same certificate construction, but their verification entry points have
different software dependencies.

\subsection{Exact arithmetic replay}
The exact replay uses CPython and \texttt{python-flint} 0.9.0;
the supplied environment record specifies CPython 3.14.2 on Linux
x86-64. The replay measurements below were obtained on an Intel
Xeon Gold 6430 under Linux x86-64. With CPython 3.14.2 installed,
run the following commands from the prepared repository root:
\begin{quote}\small
\begin{verbatim}
sha256sum --quiet -c SHA256SUMS
python3.14 -m venv .venv
.venv/bin/python -m pip install -r requirements.txt
.venv/bin/python -I certificates/derivative/replay.py
.venv/bin/python -I certificates/boundary/replay_boundary.py
\end{verbatim}
\end{quote}
The checksum manifest checks the packaged inputs; newly generated timing
reports are excluded from it. Each replay writes a new JSON report
beside its program. Python optimization
flags \texttt{-O} and \texttt{-OO} must not be used: the replay programs
use assertions. The isolated-mode flag \texttt{-I} does not disable them.
Each program resolves its data paths relative to its own directory, so
replay does not depend on an original absolute workspace path.

The programs use exact Python integers and FLINT integer/rational
arithmetic. Once the environment is installed, they require no numerical
solver, external polynomial cache or network access. Their CPU affinity
and resource-limit calls are Linux/Unix specific.

\begin{table}[tb]
\centering\small
\caption{Standalone exact replay measurements.
These are single warm-environment runs, not a cross-platform benchmark.}
\label{tab:replay}
\begin{tabular}{@{}lrr@{}}
\toprule
Replay & Elapsed time (s) & Peak RSS (MiB)\\
\midrule
Translation & 4.71 & 131.11\\
Boundary & 3.84 & 193.29\\
\bottomrule
\end{tabular}
\end{table}

Table~\ref{tab:replay} reports sequential replays on byte-preserving
copies of the certificate inputs. Each program restricted itself to one
allowed CPU. The observations measure verification of already constructed exact
certificates, not the cost of discovering them, and do not measure the
Lean build. Both runs completed every coefficient and matrix check
without reaching a resource cutoff.

\subsection{Building the formal project}
The formal project is pinned to Lean 4.29.0 and Mathlib revision
\path{8a178386ffc0f5fef0b77738bb5449d50efeea95}, recorded in
its dependency manifest. After installing the pinned Lean toolchain,
run the following commands from the prepared repository root:
\begin{quote}\small
\begin{verbatim}
cd lean
lake exe cache get
VASC_BUILD_BATCH_SIZE=1 LEAN_NUM_THREADS=1 ./build.sh
lake build
lake env lean Audit.lean
\end{verbatim}
\end{quote}
The build script schedules certificate modules in topological
order; the settings above select one target per batch and one Lean
thread to limit competing large jobs. Individual large native checks can
require several GiB, so the exact replay's memory use should not be used
to size a first formal build. The project also supplies the required
worker-stack setting. The final command prints the transitive axiom
dependencies of the exported theorem.

An import and axiom inspection confirmed the
final declaration and the stated axiom inventory from the existing
compiled libraries, without rebuilding the project. The archived build
records document a successful full project build, but do not provide a
controlled end-to-end cold-build benchmark.

A rebuild uses Bash, Lake, the pinned Lean toolchain and Mathlib; it does
not invoke Python, a semidefinite solver or the original JSON converters.
The provenance directory records source hashes, block order, dimensions
and generated-source mappings. Hashes identify the distributed artifacts;
the exact computations check their mathematical conditions.

\section{Discussion}\label{sec:discussion}
The proof provides an explicit exact certificate for a universal
inequality, together with a formal connection between the finite
certificate and its real-valued conclusion. Its central design choice
is to certify two polynomial sign conditions that fit a short
common-translation argument. Within this decomposition, a linear
multiplier is sufficient for the derivative certificate, whereas the
supplied boundary certificate uses a quadratic multiplier with both
positive-coefficient and positive-semidefinite parts.

The certificate construction involves several structural details.
The centered and uncentered Gram bases require different coefficient
expansions. Cyclic orbit compression requires a stabilizer correction at
the all-ones exponent. The multiplier must remain strictly positive on
nonzero boundary subfaces. The congruence test establishes the
invertibility of its change of coordinates. Each condition is explicit
in the proof and finite checking procedure.

Possible further work includes reducing the size of the finite data,
comparing alternative exact matrix certificates on the same instance,
and replacing native evaluation by another computational verification
mechanism. Such changes would require new checks of the resulting proof
objects.

\section*{Data and software availability}
The exact-certificate bundle and Lean project are publicly available at
\url{https://github.com/EonMath/Vasc-n11-cert}. The version accompanying this
article is commit \path{1aa3b2c61cdddd85e196948e15ff01b7c2b8db31}.
The prepared repository contains the finite rational data, checksum
manifest and standalone replay programs under \path{certificates/},
and the formal project under \path{lean/}. The latter includes the
Lean sources, embedded certificate data, pinned dependency manifests,
build script and axiom inspection entry point. Their manifests record
the source-to-certificate provenance. Selected original producer scripts
and an input frame are included under \path{generation/} for inspection;
they were syntax checked, not rerun during packaging. A fully pinned
numerical discovery pipeline and portable JSON-to-Lean converters are
not included. The package supports verification of the supplied exact
data and formal sources, rather than end-to-end regeneration of the
search or formal source files. Detailed instructions and the generation
inventory are supplied in \path{reproduce.md} and
\path{generation_inventory.md}. A fresh clean build of this exported
repository has not been performed during package preparation.

\section*{Use of AI-assisted tools}
The Lean formalization was developed using FL-Prover of the MechMath
Agent Team (MMAT). MMAT's NL-Prover was used for all other work,
including the informal proof, computational certificates, literature
research and manuscript preparation. The mathematical evidence consists
of the explicit certificates, exact checks and formal theorem.

\bibliographystyle{KLMM/klmm}
\bibliography{refs}

@article{KaltofenLiYangZhi2012,
  author = {Kaltofen, Erich L. and Li, Bin and Yang, Zhengfeng and Zhi, Lihong},
  title = {Exact certification in global polynomial optimization via sums-of-squares of rational functions with rational coefficients},
  journal = {Journal of Symbolic Computation},
  year = {2012}, volume = {47}, number = {1}, pages = {1--15},
  doi = {10.1016/j.jsc.2011.08.002}
}

@article{MagronSafeyElDin2021,
  author = {Magron, Victor and Safey El Din, Mohab},
  title = {On exact {Reznick}, {Hilbert--Artin} and {Putinar}'s representations},
  journal = {Journal of Symbolic Computation},
  year = {2021}, volume = {107}, pages = {221--250},
  doi = {10.1016/j.jsc.2021.03.005}
}

@article{DavisPapp2024,
  author = {Davis, Maria M. and Papp, D{\'a}vid},
  title = {Rational dual certificates for weighted sums-of-squares polynomials with boundable bit size},
  journal = {Journal of Symbolic Computation},
  year = {2024}, volume = {121}, pages = {102254},
  doi = {10.1016/j.jsc.2023.102254}
}

@misc{XuYu2026,
  author = {Xu, Tiancheng and Yu, Wensheng},
  title = {A Finite Certificate for the Positive {$n=11$} {Vasc} Inequality},
  year = {2026}, eprint = {2610.04549}, archivePrefix = {arXiv}, primaryClass = {cs.SC},
  url = {https://arxiv.org/abs/2610.04549}
}

@misc{GuoQiuCaoFeng2026,
  author = {Guo, Dakai and Qiu, Ruichen and Cao, Yichuan and Feng, Ruyong},
  title = {A Finite Certificate for the Positive {$n=9$} {Vasc} Inequality},
  year = {2026}, eprint = {2606.06136}, archivePrefix = {arXiv}, primaryClass = {cs.SC},
  url = {https://arxiv.org/abs/2606.06136}
}

@inproceedings{deMouraUllrich2021Lean4,
  author    = {de Moura, Leonardo and Ullrich, Sebastian},
  title     = {The {Lean 4} Theorem Prover and Programming Language},
  booktitle = {Automated Deduction -- CADE 28},
  series    = {Lecture Notes in Computer Science},
  volume    = {12699},
  pages     = {625--635},
  publisher = {Springer},
  address   = {Cham},
  year      = {2021},
  doi       = {10.1007/978-3-030-79876-5_37},
}

@inproceedings{mathlib2020,
  author    = {{The mathlib Community}},
  title     = {The {L}ean {M}athematical {L}ibrary},
  booktitle = {Proceedings of the 9th {ACM} {SIGPLAN} International Conference
               on Certified Programs and Proofs},
  series    = {CPP 2020},
  publisher = {ACM},
  address   = {New Orleans, LA, USA},
  year      = {2020},
  month     = jan,
  doi       = {10.1145/3372885.3373824}
}
\end{document}